\documentclass[11pt]{article}
\usepackage{graphicx}
\usepackage{tabularx}
\usepackage{url}
\usepackage{amsthm}
\usepackage{amsmath}
\usepackage{amsfonts}
\usepackage{amssymb}

\theoremstyle{plain}
\newtheorem{theorem}{Theorem}
\newtheorem{lemma}{Lemma}

\theoremstyle{definition}

\theoremstyle{remark}

\date{}

\begin{document}

\title{Decision trees for binary subword-closed languages}
\author{Mikhail Moshkov\thanks{Computer, Electrical and Mathematical Sciences and Engineering Division,
King Abdullah University of Science and Technology (KAUST),
Thuwal 23955-6900, Saudi Arabia. Email: mikhail.moshkov@kaust.edu.sa.
}}
\maketitle

\begin{abstract}
In this paper, we study arbitrary subword-closed languages over the alphabet
$\{0,1\}$ (binary subword-closed languages). For the set of words $L(n)$ of
the length $n$ belonging to a binary subword-closed language $L$, we
investigate the depth of decision trees solving the recognition and the
membership problems deterministically and nondeterministically. In the case
of recognition problem, for a given word from $L(n)$, we should recognize it
using queries each of which, for some $i\in \{1,\ldots ,n\}$, returns the $i$th letter
of the word. In the case of
membership problem, for a given word over the alphabet $\{0,1\}$ of the
length $n$, we should recognize if it belongs to the set $L(n)$ using the same
queries. With the growth of $n$, the minimum depth of decision trees solving
the problem of recognition deterministically is either bounded from above by
a constant, or grows as a logarithm, or linearly. For other types of trees
and problems (decision trees solving the problem of recognition
nondeterministically, and decision trees solving the membership problem
deterministically and nondeterministically), with the growth of $n$, the
minimum depth of decision trees is either bounded from above by a constant
or grows linearly. We study joint behavior of minimum depths of the
considered four types of decision trees and describe five complexity classes
of binary subword-closed languages.
\end{abstract}


\section{Introduction}
\label{sect1}

In this paper, we study arbitrary binary languages (languages over the
alphabet $E=\{0,1\}$) that are subword-closed: if a word $%
w_{1}u_{1}w_{2}\cdots w_{m}u_{m}w_{m+1}$ belongs to a language, then the
word $u_{1}\cdots u_{m}$ belongs to this language \cite%
{Brzozowski14,Haines67,Okhotin10}.

For the set of words $L(n)$ of the length $n$ belonging to a binary
subword-closed language $L$, we investigate the depth of decision trees
solving the recognition and the membership problems deterministically and
nondeterministically. In the case of recognition problem, for a given word
from $L(n)$, we should recognize it using queries each of which, for some $i\in \{1,\ldots ,n\}$, returns the $i$th letter
of the word. In the case of membership problem, for
a given word over the alphabet $E$ of the length $n$, we should recognize if
it belongs to $L(n)$ using the same queries.

For an arbitrary binary subword-closed language, with the growth of $n$, the
minimum depth of decision trees solving the problem of recognition
deterministically is either bounded from above by a constant, or grows as a
logarithm, or linearly. For other types of trees and problems (decision
trees solving the problem of recognition nondeterministically, and decision
trees solving the membership problem deterministically and
nondeterministically), with the growth of $n$, the minimum depth of decision
trees is either bounded from above by a constant, or grows linearly. We
study joint behavior of minimum depths of the considered four types of
decision trees and describe five complexity classes of binary subword-closed
languages.

In \cite{Moshkov97}, the following results were announced without proofs.
For an arbitrary regular language, with the growth of $n$, (i) the minimum
depth of decision trees solving the problem of recognition deterministically
is either bounded from above by a constant, or grows as a logarithm, or
linearly, and (ii) the minimum depth of decision trees solving the problem of
recognition nondeterministically is either bounded from above by a constant,
or grows linearly. Proofs for the case of decision trees solving the problem
of recognition deterministically can be found in \cite{Moshkov00,Moshkov05}.
To apply the considered results to a given regular language, it is necessary to know
a deterministic finite automaton (DFA) accepting this language.

Each subword-closed language over a finite alphabet is a regular language
\cite{Haines67}. In this paper, we do not assume that binary subword-closed languages are
given by DFAs. So we cannot use the results from \cite{Moshkov97,Moshkov00,Moshkov05}. Instead of
this, for binary subword-closed languages, we describe simple criteria for
the behavior of minimum depths of decision trees solving the problems of
recognition and membership deterministically and nondeterministically.

The rest of the paper is organized as follows. In Section \ref{sect2}, we
consider main notions, in Section \ref{sect3} -- main results, and in
Section \ref{sect4} -- proofs.

\section{Main Notions}
\label{sect2}

Let $\omega =\{0,1,2,\ldots \}$ be the set of nonnegative integers and $%
E=\{0,1\}$. By $E^{\ast }$ we denote the set of all finite words over the
alphabet $E$, including the empty word $\lambda $. Any subset $L$ of the set
$E^{\ast }$ is called a binary language. This language is called
subword-closed if, for any word $w_{1}u_{1}w_{2}\cdots w_{m}u_{m}w_{m+1}$
belonging to $L$, the word $u_{1}\cdots u_{m}$ belongs to $L$, where $w_{i}$%
, $u_{j}$ $\in E^{\ast }$, $i=1,\ldots ,m+1$, $j=1,\ldots ,m$. For any
natural $n$, we denote by $L(n)$ the set of words from $L$, which length is
equal to $n$. We consider two problems related to the set $L(n)$. The
problem of recognition: for a given word from $L(n)$, we should recognize it
using attributes (queries) $l_{1}^{n},\ldots ,l_{n}^{n}$, where $l_{i}^{n}$, $i\in
\{1,\ldots ,n\}$, is a function from $E^{\ast }(n)$ to $E$ such that $%
l_{i}^{n}(a_{1}\cdots a_{n})=a_{i}$ for any word $a_{1}\cdots a_{n}\in
E^{\ast }(n)$. The problem of membership: for a given word from $E^{\ast
}(n) $, we should recognize if this word belongs to the set $L(n)$ using the
same attributes. To solve these problems, we use decision trees over $L(n)$.

A decision tree over $L(n)$ is a marked finite directed tree with root,
which has the following properties:

\begin{itemize}
\item The root and the edges leaving the root are not labeled.

\item Each node, which is not the root nor terminal node, is labeled with an
attribute from the set $\{l_{1}^{n},\ldots ,l_{n}^{n}\}$.

\item Each edge leaving a node, which is not a root, is labeled with a
number from $E$.
\end{itemize}

A decision tree over $L(n)$ is called deterministic if it satisfies the
following conditions:

\begin{itemize}
\item Exactly one edge leaves the root.

\item For any node, which is not the root nor terminal node, the edges
leaving this node are labeled with pairwise different numbers.
\end{itemize}

Let $\Gamma $ be a decision tree over $L(n)$. \textit{A} complete path in $%
\Gamma $ is any sequence $\xi =v_{0},e_{0},\ldots ,v_{m},$ $e_{m},v_{m+1}$ of
nodes and edges of $\Gamma $ such that $v_{0}$ is the root, $v_{m+1}$ is a
terminal node, and $v_{i}$ is the initial and $v_{i+1}$ is the terminal
node of the edge $e_{i}$ for $i=0,\ldots ,m$. We define a subset $E(n,\xi
) $ of the set $E^{\ast }(n)$ in the following way: if $m=0$, then $E(n,\xi
)=E^{\ast }(n)$. Let $m>0$, the attribute $l_{i_{j}}^{n}$ be assigned to the
node $v_{j}$ and $b_{j}$ be the number assigned to the edge $e_{j}$, $%
j=1,\ldots ,m$. Then
\[
E(n,\xi )=\{a_{1}\cdots a_{n}\in E^{\ast }(n):a_{i_{1}}=b_{1},\ldots
,a_{i_{m}}=b_{m}\}.
\]

Let $L(n)\neq \emptyset $. We say that a decision tree $\Gamma $ over $L(n)$
solves the problem of recognition for $L(n)$ nondeterministically if $\Gamma
$ satisfies the following conditions:

\begin{itemize}
\item Each terminal node of $\Gamma $ is labeled with a word from $L(n)$.

\item For any word $w\in L(n)$, there exists a complete path $\xi $ in the
tree $\Gamma $ such that $w\in E(n,\xi )$.

\item For any word $w\in L(n)$ and for any complete path $\xi $ in the tree $%
\Gamma $ such that $w\in E(n,\xi )$, the terminal node of the path $\xi $ is
labeled with the word $w$.
\end{itemize}

We say that a decision tree $\Gamma $ over $L(n)$ solves the problem of
recognition for $L(n)$ deterministically if $\Gamma $ is a deterministic
decision tree, which solves the problem of recognition for $L(n)$
nondeterministically.

We say that a decision tree $\Gamma $ over $L(n)$ solves the problem of
membership for $L(n)$ nondeterministically if $\Gamma $ satisfies the
following conditions:

\begin{itemize}
\item Each terminal node of $\Gamma $ is labeled with a number from $E$.

\item For any word $w\in E^{\ast }(n)$, there exists a complete path $\xi $
in the tree $\Gamma $ such that $w\in E(n,\xi )$.

\item For any word $w\in E^{\ast }(n)$ and for any complete path $\xi $ in
the tree $\Gamma $ such that $w\in E(n,\xi )$, the terminal node of the path
$\xi $ is labeled with the number $1$ if $w\in L(n)$ and with the number $0$%
, otherwise.
\end{itemize}

We say that a decision tree $\Gamma $ over $L(n)$ solves the problem of
membership for $L(n)$ deterministically if $\Gamma $ is a deterministic
decision tree which solves the problem of membership for $L(n)$
nondeterministically.

Let $\Gamma $ be a decision tree over $L(n)$. We denote by $h(\Gamma )$ the
maximum number of nodes in a complete path in $\Gamma $ that are not the
root nor terminal node. The value $h(\Gamma )$ is called the depth of the
decision tree $\Gamma $.

We denote by $h_{L}^{ra}(n)$ ($h_{L}^{rd}(n)$) the minimum depth of a
decision tree, which solves the problem of recognition for $L(n)$
nondeterministically (deterministically). If $L(n)=\emptyset $, then $%
h_{L}^{ra}(n)=h_{L}^{rd}(n)=0$.

We denote by $h_{L}^{ma}(n)$ ($h_{L}^{md}(n)$) the minimum depth of a
decision tree, which solves the problem of membership for $L(n)$
nondeterministically (deterministically). If $L(n)=\emptyset $, then $%
h_{L}^{ma}(n)=h_{L}^{md}(n)=0$.

\section{Main Results}
\label{sect3}

Let $L$ be a binary subword-closed language. For any $a\in E$ and $i\in
\omega $, we denote by $a^{i}$ the word $a\cdots a$ of the length $i$ (if $%
i=0$, then $a^{i}=\lambda $). For any $a\in E\,$, let $\bar{a}=1$ if $a=0$
and $\bar{a}=0$ if $a=1$.

We define the parameter $Hom(L)$ of the language $L$, which is called the
homogeneity dimension of the language $L$. If for each natural number $m$,
there exists $a\in E$ such that the word $a^{m}\bar{a}a^{m}$ belongs to $L$,
then $Hom(L)=\infty $. Otherwise, $Hom(L)$ is the maximum number $m\in
\omega $ such that there exists $a\in E$ for which the word $a^{m}\bar{a}%
a^{m}$ belongs to $L$. If $L=\emptyset $, then $Hom(L)=0$.

We now define the parameter $Het(L)$ of the language $L$, which is called
the heterogeneity dimension of the language $L$. If for each natural number $%
m$, there exists $a\in E$ such that the word $a^{m}\bar{a}^{m}$ belongs to $%
L $, then $Het(L)=\infty $. Otherwise, $Het(L)$ is the maximum number $m\in
\omega $ such that there exists $a\in E$ for which the word $a^{m}\bar{a}%
^{m} $ belongs to $L$. If $L=\emptyset $, then $Het(L)=0$.

\begin{theorem}
\label{T1} Let $L$ be a binary subword-closed language.

(a) If $Hom(L)=\infty $, then $h_{L}^{rd}(n)=\Theta (n)$ and $%
h_{L}^{ra}(n)=\Theta (n)$.

(b) If $Hom(L)<\infty $ and $Het(L)=\infty $, then $h_{L}^{rd}(n)=\Theta
(\log n)$ and $h_{L}^{ra}(n)=O(1)$.

(c) If $Hom(L)<\infty $ and $Het(L)<\infty $, then $h_{L}^{rd}(n)=O(1)$ and $%
h_{L}^{ra}(n)=O(1)$.
\end{theorem}

For a binary subword-closed language $L$ we denote by $L^{C}$ its
complementary language $E^{\ast }\setminus L$. The notation $|L|=\infty $
means that $L$ is an infinite language, and the notation $|L|<\infty $ means
that $L$ is a finite language.

\begin{theorem}
\label{T2} Let $L$ be a binary subword-closed language.

(a) If $|L|=\infty $ and $L^{C}\neq \emptyset $, then $h_{L}^{md}(n)=\Theta
(n)$ and $h_{L}^{ma}(n)=\Theta (n)$.

(b) If $|L|<\infty $ or $L^{C}=\emptyset $, then $h_{L}^{md}(n)=O(1)$ and $%
h_{L}^{ma}(n)=O(1)$.
\end{theorem}

To study all possible types of joint behavior of functions $h_{L}^{rd}(n)$, $%
h_{L}^{ra}(n)$, $h_{L}^{md}(n)$, and $h_{L}^{ma}(n)$ for binary
subword-closed languages $L$, we consider five classes of languages $%
\mathcal{L}_{1},\ldots ,\mathcal{L}_{5}$ described in the columns 2--5 of
Table \ref{tab1}. In particular, $\mathcal{L}_{1}$ consists of all binary
subword-closed languages $L$ with $Hom(L)=\infty \ $and $L^{C}\neq
\emptyset $. It is easy to show that the complexity classes $\mathcal{L}_{1},\ldots ,%
\mathcal{L}_{5}$ are pairwise disjoint, and each binary subword-closed
language belongs to one of these classes. The behavior of functions $%
h_{L}^{rd}(n)$, $h_{L}^{ra}(n)$, $h_{L}^{md}(n)$, and $h_{L}^{ma}(n)$ for
languages from these classes is described in the last four
columns of  Table \ref{tab1}. For each class, the results considered in Table \ref{tab1} follow  from Theorems \ref{T1} and \ref{T2}, and the following three  remarks: (i) from the
condition $Hom(L)=\infty $ it follows $|L|= \infty$, (ii) from the
condition $Het(L)=\infty $ it follows $|L|= \infty$, and
(iii) from the
condition $Hom(L)<\infty $ it follows $L^{C}\neq \emptyset $.

\begin{table}[h]
\caption{Joint behavior of functions $h_{L}^{rd}$, $%
h_{L}^{ra}$, $h_{L}^{md}$, and $h_{L}^{ma}$ for binary
subword-closed languages}
\label{tab1}\center
\begin{tabular}{|l|llll|llll|}
\hline
& $Hom(L)$ & $Het(L)$ & $|L|$ & $L^{C}$ & $h_{L}^{rd}$ & $h_{L}^{ra}$
& $h_{L}^{md}$ & $h_{L}^{ma}$ \\ \hline
$\mathcal{L}_{1}$ & $=\infty $ &  &  & $\neq \emptyset $ & $\Theta (n)$ & $%
\Theta (n)$ & $\Theta (n)$ & $\Theta (n)$ \\
$\mathcal{L}_{2}$ & $=\infty $ &  &  & $=\emptyset $ & $\Theta (n)$ & $%
\Theta (n)$ & $O(1)$ & $O(1)$ \\
$\mathcal{L}_{3}$ & $<\infty $ & $=\infty $ &  &  & $\Theta (\log n)$ & $%
O(1) $ & $\Theta (n)$ & $\Theta (n)$ \\
$\mathcal{L}_{4}$ & $<\infty $ & $<\infty $ & $=\infty $ &  & $O(1)$ & $O(1)$
& $\Theta (n)$ & $\Theta (n)$ \\
$\mathcal{L}_{5}$ & $<\infty $ & $<\infty $ & $<\infty $ &  & $O(1)$ & $O(1)$
& $O(1)$ & $O(1)$\\ \hline%
\end{tabular}
\end{table}

We now show that the classes $\mathcal{L}_{1},\ldots ,\mathcal{L}_{5}$ are
nonempty. To this end, we consider the following five binary subword-closed
languages:%
\begin{eqnarray*}
L_{1} &=&\{0^{i}10^{j},0^{i}:i,j\in \omega \}, \\
L_{2} &=&E^{\ast }, \\
L_{3} &=&\{0^{i}1^{j}:i,j\in \omega \}, \\
L_{4} &=&\{0^{i}:i\in \omega \}, \\
L_{5} &=&\{0\}.
\end{eqnarray*}

It is easy to see that $L_{i}\in \mathcal{L}_{i}$ for $i=1,\ldots ,5$.

\section{Proofs of Theorems \protect\ref{T1} and \protect\ref{T2}}
\label{sect4}

In this section, we prove Theorems \ref{T1} and \ref{T2}. First, we consider
two auxiliary statements. For a word $w,$ we denote by $|w|$ its length.

\begin{lemma}
\label{L1} Let $L$ be a binary subword-closed language for which $%
Hom(L)<\infty $. Then any word $w$ from $L$ can be represented in the form
\begin{equation}
w_{1}a^{i}w_{2}\bar{a}^{j}w_{3},  \label{E}
\end{equation}%
where $a\in E$, $i,j\in \omega $, and $w_{1}$, $w_{2}$, $w_{3}$ are words
from $E^{\ast }$ with length at most $2Hom(L)$ each.
\end{lemma}

\begin{proof}
Denote $m=$ $Hom(L)$. Then the words $0^{m+1}10^{m+1}$ and $1^{m+1}01^{m+1}$
do not belong to $L$. Let $w$ be a word from $L$. Then, for any $a\in E$,
any entry of the letter $a$ in $w$ has at most $m$ $\bar{a}$s to the left of
this entry (we call it $l$-entry of $a$) or at most $m$ $\bar{a}$s to the
right of this entry (we call it $r$-entry of $a$). Let $a\in E$. We say that
$w$ is (i) $a$-$l$-word if any entry of $a$ in $w$ is $l$-entry; (ii) $a$-$r$%
-word if any entry of $a$ in $w$ is $r$-entry; and (iii) $a$-$b$-word if $w$
is not $a$-$l$-word and is not $a$-$r$-word. Let $c,d\in \{l,r,b\}$. We say
that $w$ is $cd$-word if $w$ is $0$-$c$-word and $1$-$d$-word. There are
nine possible pairs $cd$. We divide them into four groups: (a) $ll$ and $rr$%
, (b) $lr$ and $rl$, (c) $lb$, $rb$, $bl$, and $br$, and (d) $bb$, and
consider them separately. Let $$w=a_{1}\cdots a_{n}.$$ We assume that $w$
contains both 0s and 1s. Otherwise, $w$ can be represented in the form (\ref%
{E}).

(a) Let $w$ be $ll$-word. Let $a_{n}=0$ and $a_{i}$ be the rightmost entry
of $1$ in $w$. Since $w$ is $ll$-word, there are at most $m$ 1s to the left
of $a_{n}$ and at most $m$ 0s to the left of $a_{i}$. Denote $%
w_{1}=a_{1}\cdots a_{i}$. Then $w_{1}$ contains at most $m$ 0s and at most $m
$ 1s, i.e., the length of $w_{1}$ is at most $2m$. Moreover, to the right of
$a_{i}$ there are only 0s. Thus, $w=w_{1}0^{n-i}$, where $|w_{1}|=i\leq 2m$,
i.e, $w$ can be represented in the form (\ref{E}).

Let $a_{n}=1$ and $a_{i}$ be the rightmost entry of $0$ in $w$. Denote $%
w_{1}=a_{1}\cdots a_{i}$. Then $w_{1}$ contains at most $m$ 0s and at most $m
$ 1s, i.e., $|w_{1}|\leq 2m$. Moreover, to the right of $a_{i}$ there are
only 1s. Thus, $w=w_{1}1^{n-i}$, i.e, $w$ can be represented in the form (%
\ref{E}).

One can prove in a similar way that any $rr$-word can be represented in the
form (\ref{E}).

(b) Let $w$ be $lr$-word, $a_{i}$ be the rightmost entry of 0 and $a_{j}$ be
the leftmost entry of 1. Then either $j=i+1$ or $j<i$. Let $j=i+1$. Then $%
w=0^{i}1^{n-i}$, i.e., $w$ can be represented in the form (\ref{E}). Let now
$j<i$. Denote $w_{2}=a_{j}\cdots a_{i}$. The word $w$ has at most $m$ 0s to
the right of $a_{j}$ and at most $m$ 1s to the left of $a_{i}$. Therefore  $%
|w_{2}|\leq 2m$ and $w=0^{j-1}w_{2}1^{n-i}$, i.e., $w$ can be represented in
the form (\ref{E}).

One can prove in a similar way that any $rl$-word can be represented in the
form (\ref{E}).

(c) Let $w$ be $lb$-word, $a_{i}$ be the rightmost entry of 1 such that to
the left of this entry we have at most $m$ 0s and $a_{j}$ be the next after $%
a_{i}$ entry of 1. It is clear that to the right of $a_{j}$ there are at
most $m$ 0s, $j\geq i+2$, and all letters $a_{i+1},\ldots ,a_{j-1}$ are
equal to 0. Let $a_{k}$ be the rightmost entry of 0. Then to the left of $a_{k}$
there are at most $m$ 1s. It is clear that either $k=j-1$ or $k>j$. Denote $%
w_{1}=a_{1}\cdots a_{i}$. Then $|w_{1}|\leq 2m$. Let $k=j-1$. In this case, $%
w=w_{1}0^{j-i-1}1^{n-j+1}$, i.e., $w$ can be represented in the form (\ref{E}).
Let $k>j$. Denote $w_{2}=a_{j}\cdots a_{k}$. Then $|w_{2}|\leq 2m$. We have $%
w=w_{1}0^{j-i-1}w_{2}1^{n-k}$, i.e., $w$ can be represented in the form (\ref%
{E}).

One can prove in a similar way that any $rb$- or $bl$-, or $br$-word can be
represented in the form (\ref{E}).

(d) Let $w$ be $bb$-word, $a_{i}$ be the rightmost entry of 0 such that
there are at most $m$ 1s to the left of this entry and $a_{j}$ be the next
after $a_{i}$ entry of 0. Then there are at most $m$ 1s to the right of $%
a_{j}$, $j\geq i+2$, and $w=a_{1}\cdots a_{i}1\cdots 1a_{j}\cdots a_{n}$.
Denote $A=\{1,\ldots ,i\}$, $B=\{i+1,\ldots ,j-1\}$, and $C=\{j,\ldots ,n\}$%
. Let $a_{k}$ be the rightmost entry of 1 such that there are at most $m$ 0s
to the left of this entry and $a_{l}$ be the next after $a_{k}$ entry of 1.
Then there are at most $m$ 0s to the right of $a_{l}$, $l\geq k+2$, and $%
w=a_{1}\cdots a_{k}0\cdots 0a_{l}\cdots a_{n}$.

There are four possible types of
location of $a_{k}$ and $a_{l}$: (i) $k\in A$ and $l\in A$, (ii) $k\in A$
and $l\in B$ (the combination $k\in A$ and $l\in C$ is impossible since all
letters with indices from $B$ are 1s but all letters between $a_{k}$ and $%
a_{l}$ are 0s), (iii) $k\in B$ and $l\in C$ (the combination $k\in B$ and $%
l\in B$ is impossible since all letters with indices from $B$ are 1s but all
letters between $a_{k}$ and $a_{l}$ are 0s), and (iv) $k\in C$ and $l\in C$. We now consider cases (i)--(iv) in detail.

(i) Let $k\in A$ and $l\in A$. Then $w=a_{1}\cdots a_{k}0\cdots 0a_{l}\cdots
a_{i}1\cdots 1a_{j}\cdots a_{n}$. Denote $w_{1}=a_{1}\cdots a_{k}$, $%
w_{2}=a_{l}\cdots a_{i}$, and $w_{3}=a_{j}\cdots a_{n}$. The length of $w_{1}
$ is at most $2m$ since from the left of $a_{k}$ there are at most $m$ 0s
and from the left of $a_{i}$ there are at most $m$ 1s. We can prove in a
similar way that $|w_{2}|\leq 2m$ and $|w_{3}|\leq 2m$. Therefore $w$ can be
represented in the form (\ref{E}).

(ii) Let $k\in A$ and $l\in B$. Then $l=i+1$ and $$w=a_{1}\cdots a_{k}0\cdots
0a_{i}a_{i+1}1\cdots 1a_{j}\cdots a_{n},$$ where $a_{i}=0$ and $a_{i+1}=1$.
Denote $w_{1}=a_{1}\cdots a_{k}$ and $w_{3}=a_{j}\cdots a_{n}$. It is easy
to show that $|w_{1}|\leq 2m$ and $|w_{3}|\leq 2m$. Therefore $w$ can be
represented in the form (\ref{E}).

(iii) Let $k\in B$ and $l\in C$. Then $k=j-1$ and $$w=a_{1}\cdots
a_{i}1\cdots 1a_{j-1}a_{j}0\cdots 0a_{l}\cdots a_{n},$$ where $a_{j-1}=1$ and
$a_{j}=0$. Denote $w_{1}=a_{1}\cdots a_{i}$ and $w_{3}=a_{l}\cdots a_{n}$.
It is easy to show that $|w_{1}|\leq 2m$ and $|w_{3}|\leq 2m$. Therefore $w$
can be represented in the form (\ref{E}).

(iv) Let $k\in C$ and $l\in C$. Then $w=a_{1}\cdots a_{i}1\cdots
1a_{j}\cdots a_{k}0\cdots 0a_{l}\cdots a_{n}$. Denote $w_{1}=a_{1}\cdots
a_{i}$, $w_{2}=a_{j}\cdots a_{k}$, and $w_{3}=a_{l}\cdots a_{n}$. It is easy
to show that $|w_{1}|\leq 2m$, $|w_{2}|\leq 2m$, and $|w_{3}|\leq 2m$.
Therefore $w$ can be represented in the form (\ref{E}).
\end{proof}

\begin{lemma}
\label{L2} Let $L$ be a binary subword-closed language for which $%
Hom(L)<\infty $ and $Het(L)<\infty $. Then there exists natural $p$ such
that $|L(n)|\leq p$ for any natural $n$.
\end{lemma}

\begin{proof}
Denote $m=\max (Hom(L),Het(L))$. Then the words $0^{m+1}1^{m+1}$ and $1^{m+1}0^{m+1}$ do not belong to $L$. Using Lemma \ref{L1}, we obtain that each word $w$ from $L$ can be
represented in the form $w_{1}a^{i}w_{2}\bar{a}^{j}w_{3}$, where $a\in E$,
the length of $w_{k}$ is at most $t=2m$ for $k=1,2,3$, $i,j\in \omega $, and $%
i\leq m$ or $j\leq m$. We now evaluate the number of such words, which
length is equal to $n$. Let $k\in \{1,2,3\}$. Then the number of different
words $w_{k}$ is at most $2^{0}+2^{1}+\cdots +2^{t}<2^{t+1}$. Let us assume that the words $w_1$, $w_2$, and $w_3$ are fixed and $|w_1|+|w_2|+|w_3| \le n$. Then the number of
different words $a^{i}\bar{a}^{j}$ of the length $n-|w_1|-|w_2|-|w_3|$ is at most $4(m+1)$ since $i\leq m$ or $%
j\leq m$. Thus, the number of words in $L(n)$ is at most $p=2^{3t+3}(2t+4)$.
\end{proof}

 \begin{proof}[Proof of Theorem \protect\ref{T1}]
It is clear that $h_{L}^{ra}(n)\leq h_{L}^{rd}(n)$ for any
natural $n$.

(a) Let $Hom(L)=\infty $ and $n$ be a natural number. Then there exists $%
a\in E$ such that $a^{n}\bar{a}a^{n}\in L$. Therefore $a^{n},a^{i}\bar{a}%
a^{n-i-1}\in L(n)$ for $i=0,\ldots ,n-1$. Let $\Gamma $ be a decision tree
over $L(n)$, which solves the problem of recognition for $L(n)$
nondeterministically and has the minimum depth $h_{L}^{ra}(n)$, and $\xi $
be a complete path in $\Gamma $ such that $a^{n}\in E(n,\xi )$. Let us
assume that there is $i\in \{0,\ldots ,n-1\}$ such that the attribute $%
l_{i+1}^n$ is not attached to any node of $\xi $, which is not the root nor
the terminal node. Then $a^{i}\bar{a}a^{n-i-1}\in E(n,\xi )$, which is
impossible. Therefore $h(\Gamma )\geq n$ and $h_{L}^{ra}(n)\geq n$. It is
easy to show that $h_{L}^{rd}(n)\leq n$. Thus, $%
h_{L}^{ra}(n)=h_{L}^{rd}(n)=n $ for any natural $n$.

(b) Let $Hom(L)<\infty $ and $Het(L)=\infty $. By Lemma \ref{L1}, each word
from $L$ can be represented in the form $w_{1}a^{i}w_{2}\bar{a}^{j}w_{3}$%
, where $a\in E$, the length of $w_{k}$ is at most $t=2Hom(L)$ for $k=1,2,3$%
, and $i,j\in \omega $. Note that either $w_{2}=\lambda$ or $w_{2}$ is a word of the kind
 $\bar{a}\cdots a$.

Let $n$ be a natural number such that $n\geq 10t$.  We now describe the work of a decision tree
over $L(n)$, which solves the problem of recognition for $L(n)$
deterministically.
Let $w\in L(n)$. We represent this word as follows: $%
w=L_{1}L_{2}L_{3}AR_{3}R_{2}R_{1}$, where the length of each word $%
L_{1},L_{2},L_{3},R_{3},R_{2},R_{1}$ is equal to $t$. First, we recognize
all letters in the words $L_{1},L_{2},R_{2},R_{1}$ using $4t$ queries (attributes). We now consider four cases.

(i) Let $L_{2}=R_{2}=a^{t}$ for some $a\in E$. Then $L_{3}AR_{3}=a^{n-4t}$
and the word $w$ is recognized.

(ii) Let $L_{2}=a^{t}$ for some $a\in E$ and $R_{2}$ contain both 0 and 1.
Then $R_{2}$ has an intersection with the word $w_{2}$. It is clear that $w_2$ has no intersection with the word $A$ and  $L_{3}A=a^{n-5t}$. We recognize all
letters of the word $R_{3}$. As a result, the word $w$ will be recognized.

(iii) Let $R_{2}=a^{t}$ for some $a\in E$ and $L_{2}$ contain both 0 and 1.
Then $L_{2}$ has an intersection with the word $w_{2}$. It is clear that $w_2$ has no intersection with the word $A$ and  $AR_{3}=a^{n-5t}$. We recognize all
letters of the word $L_{3}$. As a result, the word $w$ will be recognized.

(iv) Let $L_{2}=a^{t}$ and $R_{2}=\bar{a}^{t}$ for some $a\in E$. Then we
need to recognize the position of the word $w_{2}$ and the word $w_{2}$ itself.
Beginning with the left, we divide $L_{3}AR_{3}$ and, probably, a prefix of $%
R_{2}$ into blocks of the length $t$. As a result, we  have $k\leq n/t$
blocks. We recognize all letters in the block with number $r=\left\lceil
k/2\right\rceil $. If all letters in this block are equal to $\bar{a}$, then we
apply the same procedure to the blocks with numbers $1,\ldots ,r-1$. If all
letters in this block are equal to $a$, then we  apply the same
procedure to the blocks with numbers $r+1,\ldots ,k$. If the considered
block contains both 0 and 1, then we recognize $t$ letters before this block
and $t$ letters after this block and, as a result, recognize both the word $%
w_{2}$ and its position. After each iteration, the number of blocks is at
most one half of the previous number of blocks. Let $q$ be the whole number of iterations. Then after the iteration $q-1$ we have at least one unchecked block. Therefore $k/2^{q-1}\ge 1$ and $q \le \log _{2}k +1$.

In the case (i), to recognize the word $w$ we make $4t$ queries. In the
cases (ii) and (iii), we make $5t$ queries. In the case (iv), we make at
most $t\log _{2}(n/t) +7t$ queries. As a result,
we have $h_{L}^{rd}(n)=O(\log n)$.

Since $Het(L)=\infty $, for any natural $n$, the set $L(n)$ contains for
some $a\in E$ words $a^{i}\bar{a}^{n-i}$ for $i=0,\ldots ,n$. Then $%
|L(n)|\geq n+1$, and each decision tree $\Gamma $ over $L(n)$ solving the
problem of recognition for $L(n)$ deterministically has at least $n+1
$ terminal nodes. One can show that the number of terminal nodes in $\Gamma $
is at most $2^{h(\Gamma )}$. Therefore $h(\Gamma )\geq \log _{2}(n+1)$. Thus
$h_{L}^{rd}(n)=\Omega (\log n)$ and $h_{L}^{rd}(n)=\Theta (\log n)$.

We now prove that $h_{L}^{ra}(n)=O(1)$. To this end, it is enough to
show that there is a natural number $c$ such that, for each natural $n$ and
for each word $w\in L(n)$, there exists a subset $B_{w}$ of the set of
attributes $\{l_{1}^{n},\ldots ,l_{n}^{n}\}$ such that $\left\vert
B_{w}\right\vert \leq c$ and, for any word $u\in L(n)$ different from $w$,
there exists an attribute $l_{i}^{n}\in B_{w}$ for which $l_{i}^{n}(w)\neq
l_{i}^{n}(u)$. We now show that as $c$ we can use the number $7t$. In
the case (i), in the capacity of the set $B_{w}$ we can choose all
attributes corresponding to $4t$ letters from the subwords $L_{1}$, $L_{2}$,
$R_{2}$, and $R_{1}$. In the case (ii), we can choose all attributes
corresponding to $5t$ letters from the subwords $L_{1}$, $L_{2}$, $R_{3}$, $%
R_{2}$, and $R_{1}$. In the case (iii), we can choose all attributes
corresponding to $5t$ letters from the subwords $L_{1}$, $L_{2}$, $L_{3}$, $%
R_{2}$, and $R_{1}$. In the case (iv), in the capacity of the set $B_{w}$ we
can choose all attributes corresponding to $4t$ letters from the subwords $%
L_{1}$, $L_{2}$, $R_{2}$, and $R_{1}$, and $3t$ letters from the block containing
both 0 and 1 and from the blocks that are its left and right neighbors.

(c) Let $Hom(L)<\infty $ and $Het(L)<\infty $. By Lemma \ref{L2}, there
exists natural $p$ such that $|L(n)|\leq p$ for any natural $n$. Let $n$ be
a natural number. Then the set $L(n)$ contains at most $p$ words, and there
exists a subset $B$ of the set of attributes $\{l_{1}^{n},\ldots ,l_{n}^{n}\}
$ such that $\left\vert B\right\vert \leq p^2$ and, for any two different
words $u,w\in L(n)$, there exists an attribute $l_{i}^{n}\in B$ for which $%
l_{i}^{n}(w)\neq l_{i}^{n}(u)$. It is easy to construct a decision tree over
$L(n)$ which solves the problem of recognition for $L(n)$ deterministically by sequential
computing attributes from $B$. The depth of this tree is at most $p^{2}$.
Therefore $h_{L}^{rd}(n)=O(1)$ and $h_{L}^{ra}(n)=O(1)$. 
\end{proof}

\begin{proof}[Proof of Theorem \protect\ref{T2}]
It is clear that $h_{L}^{ma}(n)\leq h_{L}^{md}(n)$ for any
natural $n$.

(a) Let $|L|=\infty $, $L^{C}\neq \emptyset $, and $w_{0}$ be a word with
the minimum length from $L^{C}$. Since $|L|=\infty $, $L(n)\neq \emptyset $ for any natural $n$. Let $n$ be a
natural number such that $n>|w_{0}|$ and $\Gamma $ be a decision
tree over $L(n)$ that solves the problem of membership for $L(n)$
nondeterministically and has the minimum depth. Let $w\in L(n)$ and $\xi $
be a complete path in $\Gamma $ such that $w\in E(n,\xi )$. Then the terminal
node of $\xi $ is labeled with the number $1$. Let us assume that the number
of nodes labeled with attributes in $\xi $ is at most $n-|w_{0}|$. Then we
can change at most $|w_{0}|$ letters in the word $w$ such that the
obtained word $w^{\prime }$ will satisfy the following conditions: $w_{0}$
is a subword of $w^{\prime }$ and $w^{\prime }\in $ $E(n,\xi )$. However it
is impossible since in this case $w^{\prime }\notin L(n)$ and $w^{\prime }\in $
$E(n,\xi )$ but the terminal node of $\xi $ is labeled with the number 1.
Therefore the depth of $\Gamma $ is greater than $n-|w_{0}|$. Thus $%
h_{L}^{ma}(n)=\Omega (n)$. It is easy to construct a decision tree over $L(n)
$ that solves the problem of membership for $L(n)$ deterministically and has
the depth equals to $n$. Therefore $h_{L}^{md}(n)=O(n)$. Thus, $%
h_{L}^{md}(n)=\Theta (n)$ and $h_{L}^{ma}(n)=\Theta (n)$.

(b) Let  $|L|<\infty $. Then there exists natural $m$ such that $%
L(n)=\emptyset $ for any natural $n\geq m$. Therefore, for each natural $%
n\geq m$, $h_{L}^{md}(n)=0$ and $h_{L}^{ma}(n)=0$.

Let $L^{C}=\emptyset $, $n$ be a natural number, and $\Gamma $ be a decision
tree over $L(n)$ which consists of the root, a terminal node labeled with $1,$
and an edge that leaves the root and enters the terminal node. One can show
that $\Gamma $ solves the problem of membership for $L(n)$ deterministically
and has the depth equals to $0$. Therefore $h_{L}^{md}(n)=0$ and $%
h_{L}^{ma}(n)=0$. 
\end{proof}

\subsection*{Acknowledgments}

Research reported in this publication was supported by King Abdullah
University of Science and Technology (KAUST).

\bibliographystyle{spmpsci}
\bibliography{bsc-languages}

\end{document}